\documentclass[pra,twocolumn]{revtex4}%
\usepackage{amsfonts}
\usepackage{amsmath}
\usepackage{amssymb}
\usepackage{graphicx}%
\newtheorem{theorem}{Theorem}

\newenvironment{proof}[1][Proof]{\noindent\textbf{#1.} }{\ \rule{0.5em}{0.5em}}
\begin{document}
\title{Quantum state transfer through a qubit network with energy shifts and fluctuations}
\author{Andrea Casaccino}
\affiliation{Research Laboratory of Electronics, Massachusetts Institute of
Technology, Cambridge, MA,02139, USA\\ Information Engineering Department,
University of Siena, I-53100 Siena, Italy}
\author{Seth Lloyd}
\affiliation{Research Laboratory of Electronics, Massachusetts Institute of
Technology, Cambridge, MA, 02139, USA}
\author{Stefano Mancini}
\affiliation{Department of Physics, University of Camerino, I-62032 Camerino, Italy}
\author{Simone Severini}
\affiliation{Institute for Quantum Computing and Department of Combinatorics \&
Optimization, University of Waterloo, Waterloo N2L 3G1, ON Canada}

\begin{abstract}
We study quantum state transfer through a qubit network modeled by spins with
XY interaction, when relying on a single excitation. We show that it is possible to
achieve perfect transfer by shifting (adding) energy to specific vertices. This
technique appears to be a potentially powerful tool to change, and in some
cases improve, transfer capabilities of quantum networks. Analytical results are
presented for all-to-all networks and all-to-all networks with a missing link.
Moreover, we evaluate the effect of random fluctuations on the transmission
fidelity.

\end{abstract}
\maketitle

\section{Introduction}

The rapid growth of the area of quantum information has led to consider the idea
of multi users quantum networks with the final goal of realizing a number of
nano-scale devices and communication protocols \cite{Kimble}. The study of
networks of interacting qubits (spins) constitutes a good testing ground for this
purpose. In the last few years, this kind of networks have been specifically
considered to be good candidates for engineering perfect quantum channels and
allowing information transfer between distant locations \cite{sou, ch, sub} (see
also \cite{sou1}, for a review). Such networks appear to be useful for the
implementation of data buses in quantum mechanical devices, in particular
because they undergo a  free dynamics after an initial set-up.

In this perspective, the possibility of having \emph{perfect state transfer} (for
short, \emph{PST}) comes from suitable quantum interference effects in the
network dynamics. However, one of the problems arising in such a scenario is
given by natural dispersion effects and destructive interference, which determine
a loss of information between communicating sites. In the worst cases,
information can even remain totally localized, due to Anderson localization
effects \cite{And}. While this situation may still be useful, this is not the case
when designing protocols for distant communication.

In a number of recent papers, PST has been related to the combinatorial
properties of  networks (see, \emph{e.g.}, \cite{be}, and the references contained
therein). In particular, in the $XY$ model (respectively, the $XYZ$ model), when
considering a single excitation, it has been shown that PST essentially depends
on the eigensystem of the adjacency matrix of the graph (respectively, the
Laplacian matrix), because certain invariant eigenspaces of the total Hilbert
space evolve independently.

Here we discuss the problem of how to improve the fidelity of excitation transfer
for a fixed interaction ($XY$ model) and network. In particular, we show that, by a
suitable energy shift corresponding to some vertices in the network, it is possible
to achieve perfect transfer in cases where this does not usually happen. We
conjecture that this is possible in many networks, whenever we add a suitable
amount of energy. Moreover, we evaluate the effect of random fluctuations on the
transmission fidelity. We separately consider noise affecting qubits' frequencies
and qubits' couplings and we show signatures of Anderson localization \cite{And}
as well as of stochastic resonance \cite{Fabio}.

The structure of the paper is as follows. In Section II, we describe the model
considered here. In Section III, we give rigorous results for all-to-all networks and
all-to-all networks with a missing link, therefore extending the cases studied in
\cite{ca}. For these networks, we show that a certain energy shift allows PST.
Indeed, it is well-known that there is no PST for an all-to-all network without
energy shift. For the case of an all-to-all network with a missing link, the energy
shift changes the periodicity of the evolution. In Section IV, we discuss how to
enhance the transfer fidelity for a linear spin chain. It is known that a spin chain
with constant couplings allows PST between its end-vertices only when it has
length two or three. Evidence given by numerics show that PST can be achieved
in chains of any length by an appropriate energy shift independent of the number
of nodes. The drawback is a rapid increase of the transfer time. Furthermore, the
number of geodesics between the input and output vertex seems to play a role in
determining the transfer time. Finally, in Section V, we show how noise affects
the transfer. In particular, we show that disordered couplings are more
deleterious than disordered frequencies when optimal energy shift is used.  In the
absence of such a shift,  the noise may enhance the transmission fidelity.
Conclusions are drawn in Section VI, where we briefly summarize the results and
outline potential applications.

\section{Set-up}

Let $G=(V,E)$ be a simple undirected graph (that is, without loops or parallel
edges), with set of vertices $V(G)$ (such that $|V(G)|=n$) and set of edges
$E(G)$. The \emph{adjacency matrix} of $G$ is denoted by $A(G)$ and defined by
$[A(G)]_{ij}=1$, if $ij\in E(G)$; $[A(G)]_{ij}=0$ if $ij\notin E(G)$. The
adjacency matrix is a useful tool to describe a network of $n$ spin-$1/2$
quantum particles. The particles are usually attached to the vertices of $G$,
while the edges of $G$ represent their allowed couplings. If one considers the
$XY$ interaction model then $\{i,j\}\in E(G)$ means that the particles $i$ and
$j$ interact by the Hamiltonian $[H_{XY}(G)]_{ij}=\left(  X_{i}X_{j}%
+Y_{i}Y_{j}\right)  $. Throughout the paper $X_{i}$ and $Y_{i}$ denote the
usual Pauli operators of the $i$-th particle.

Here we consider unit coupling constant. Thus, the Hamiltonian of the whole
network reads
\begin{equation}
H_{XY}(G)=\frac{1}{2}\sum_{i\neq j=1}^{n}[A(G)]_{ij}\left(  X_{i}X_{j}%
+Y_{i}Y_{j}\right)  \label{Hnet}%
\end{equation}
and it acts on the Hilbert space $\left(  \mathbb{C}^{2}\right)^{\otimes n}%
$. Let us now restrict our attention to the single excitation subspace
$\mathbb{C}^{n}$, \emph{i.e.}, the subspace of dimension $n$ spanned by the
vectors $\{|1\rangle,\ldots,|n\rangle\}$. A vector $|j\rangle$ indicates the presence
of the excitation on the $j$-th site and the absence on all the others. This is
equivalent to the following tensor product of the $Z$- eigenstates
$|\underset{n}{\underbrace{0\ldots010\ldots0}}\rangle$, being $1$ in the $j$-th
position. In the basis $\{|1\rangle,\ldots,|n\rangle\}$, the
Hamiltonian coming from Eq. (\ref{Hnet}) has entries $[H_{XY}(G)]_{ij}%
=2[A(G)]_{ij}$. This will be called the $XY${ \emph{adjacency matrix} of the
graph }$G$.
Hereafter, we shall consider the possibility of adding an amount $\Delta_{E}$
of free energy to desired sites. In this case, the $XY$ Hamiltonian reads
\begin{equation}
\lbrack H_{XY}(G,E_{i})]_{ij}=\left\{
\begin{tabular}
[c]{ll}%
$\Delta_{E}(i),$ & if $i=j;$\\
$2,$ & if $i,j\in E(G);$\\
$0,$ & otherwise,
\end{tabular}
\ \ \ \right.  \label{HGij}%
\end{equation}
We simply write $\Delta_{E}$ instead of $\Delta_{E}(i)$ when $i$ is clear from
the context. Finally, let us recall the definition of the \emph{fidelity} at time $t$
between vertex $i$ and vertex $j$ as $f_{G}(i,j;t):=|\langle i|e^{-\iota
H(G)t}|j\rangle|^{2}$, where $i$ represents the input vertex and $j$ the output
vertex (in short $I/O$).

\section{Fidelity}

In this section, we present rigorous results about the effects of an energy shift
only in the input/output vertices for two specific networks: we consider the case
of the \emph{complete graph}, $K_{n}$, and of the \emph{complete graph with a
missing link}, $K_{n}^{-}$. In these two cases, given the Hamiltonian $H_{XY}$,
we express analytically the fidelity and the transfer time as a function of $n$ and
$\Delta_{E}$.

\subsection{Complete graph}

Every two vertices of the complete graph $K_{n}$ are adjacent. For this graph,
we can prove the next result:

\begin{theorem}
\label{th1}Let $\alpha=\sqrt{4n^{2}-4(n-4)\Delta_{E}+\Delta_{E}^{2}}$ with
$n\geq4$ and $k\in\mathbb{N}$. For an energy shift $\Delta_{E}(i,j)$ on the
vertices $i,j\in I/O$, we have the following observations:

\begin{itemize}
\item $\max_{t}f_{K_{n}}(i,i;t)=\max_{t}f_{K_{n}}(j,j;t)=1$, for $\Delta
_{E}(i,j)=2n$ and $t=2k\pi/\alpha$;

\item $\max_{t}f_{K_{n}^{-}}(k,k;t)=1$, for every $k\notin I/O$ and
$t=4k\pi/\alpha$.

When $i\neq j$,

\item $\max_{t}f_{K_{n}}(i,j;t)=1$, for $\Delta_{E}(i,j)=2n$ and $t=\left(
2\pi+4\pi k\right) /\alpha$;

\item $\max_{t}f_{K_{n}}(i,k;t)=16/\alpha^{2}$, for $\Delta_{E}(i)=2n$,
$k\notin I/O$ and $t=\left(  2\pi+4\pi k\right)/\alpha$;

\item $\max_{t}f_{K_{n}}(k,l;t)=[(\alpha(n-2)-2)]^{2}/4\alpha^{2}(n-2)^{2}$,
for $k,l\notin I/O$ and $t=2k\pi/\alpha$.
\end{itemize}
\end{theorem}

\begin{proof}
{The }$XY${ adjacency matrix of $K_{n}$ has the form
\[
\lbrack H_{XY}({K_{n}})]_{ij}=\left\{
\begin{tabular}
[c]{ll}%
$\Delta_{E},$ & if $i=j\in I/O;$\\
$0,$ & if $i=j\not \in I/O;$\\
$2,$ & otherwise.
\end{tabular}
\ \ \right.
\]
}

{The characteristic polynomial $P(\lambda)$ can be obtained as a function of $n$
and $\Delta_{E}$:
\begin{align*}
P(\lambda)  &  =(\lambda+2)^{n-3}(\Delta_{E}-2-\lambda)\times\left(
4(n-1)\right. \\
&  \left.  -2(n-3)\Delta_{E}+2(n-2)\lambda+\Delta_{E}\lambda-\lambda
^{2}\right)  .
\end{align*}
The roots of $P(\lambda)$ are as follows: }$\lambda_{1}=\Delta_{E}-2$,
$\lambda_{2}^{n-3}=-2$, $\lambda_{3,4}^{\pm}=(2(n-2)+\Delta_{E}\pm\alpha)/2$.
A corresponding (unnormalized) orthogonal basis of eigenvectors can be written
as%
\begin{align*}
|\lambda_{1}\rangle &  =(-1,0,\ldots,0,1)\\
\lbrack|\lambda_{2}^{1\leq l\leq n-3}\rangle]_{u}  &  =\left\{
\begin{tabular}
[c]{rl}%
$-\frac{1}{l},$ & if $u\in\{2,n-r:1\leqslant r\leqslant l-1\};$\\
$1,$ & if $u=n-l;$\\
$0,$ & otherwise,
\end{tabular}
\ \ \right. \\
|\lambda_{3,4}^{\pm}\rangle &  =(1,\omega^{\pm},\ldots,\omega^{\pm},1),
\end{align*}
where $\omega^{\pm}=\frac{1}{4(n-2)}(2(n-4)-\Delta_{E}\pm\alpha)$. Thus, from
the spectral decomposition of the unitary matrix in the canonical basis,
$U_{t}(K_{n})\equiv e^{-\iota H(K_{n})t}$, we have the following diagonal entries:

\begin{itemize}
\item if $i\in I/O$ then%
\begin{align*}
\lbrack U_{t}(K_{n})]_{ii}  &  =\frac{1}{4\alpha}\left(  \alpha-2n+\Delta
_{E}+8\right)  e^{-\iota\left[  \lambda_{3}\right]  t}\\
&  +\frac{1}{4\alpha}\left(  \alpha-2n-\Delta_{E}+8\right)  e^{-\iota\left[
\lambda_{4}\right]  t}\\
&  +\frac{1}{2}e^{-\iota\left[  \lambda_{1}\right]  t};
\end{align*}

\item if $i\notin I/O$ then
\begin{align*}
\lbrack U_{t}(K_{n})]_{ii}  &  =\frac{1}{n-2}\left(  n-3\right)
e^{-\iota\left[  \lambda_{1}\right]  t}\\
&  +\frac{1}{2n\alpha-4\alpha}\left(  \alpha-2n+\Delta_{E}+8\right)
e^{-\iota\left[  \lambda_{3}\right]  t}\\
&  +\frac{1}{2n\alpha-4\alpha}\left(  \alpha+2n-\Delta_{E}-8\right)
e^{-\iota\left[  \lambda_{4}\right]  t}.
\end{align*}

\end{itemize}

The off-diagonal entries of $U_{t}(K_{n})$ are as follows:

\begin{itemize}
\item if $i\neq j$ and $i,j\in I/O$ then%
\begin{align*}
\lbrack U_{t}(K_{n})]_{ij}  &  =\frac{\Delta_{E}-2(n-4)+\alpha}{4\alpha
}e^{-\iota\left[  \lambda_{3}\right]  t}\\
&  +\frac{2(n-4)-\Delta_{E}+\alpha}{4\alpha}e^{-\iota\left[  \lambda
_{4}\right]  t}\\
&  -\frac{1}{2}e^{-\iota\left[  \lambda_{1}\right]  t};
\end{align*}

\item if $i\neq j$, $i\in I/O$ and $j\notin I/O$ or \emph{viz}, then
$[U_{t}(K_{n})]_{ij}=2(e^{-\iota\left[  \lambda_{3}\right]  t}-e^{-\iota
\left[  \lambda_{4}\right]  t})/\alpha$.

\item if $i\neq j$ and $i,j\notin I/O$ then%
\begin{align*}
\lbrack U_{t}(K_{n})]_{ij}  &  =\frac{\Delta_{E}-2(n-4)+\alpha}{2(n-2)\alpha
}e^{-\iota\left[  \lambda_{4}\right]  t}\\
&  +\frac{2(n-4)-\Delta_{E}+\alpha}{2(n-2)\alpha}e^{-\iota\left[  \lambda
_{3}\right]  t}\\
&  -\frac{1}{(n-2)\alpha}e^{-\iota\left[  \lambda_{2}\right]  t}.
\end{align*}

\end{itemize}

This gives us the tools to evaluate the fidelity for generic situations. For
instance, if we take $i,j\in I/O$, the fidelity $f(i,j;t)=|\langle
j|U_{t}(K_{n})|i\rangle|^{2}$ reads%
\begin{align}
f(i,j;t)  &  =\frac{\Delta_{E}^{2}+3\alpha^{2}-4\Delta_{E}(n-4)+4(n-4)^{2}%
}{8\alpha^{2}}\label{fid}\\
&  +\frac{\left(  8+\Delta_{E}+\alpha-2n\right)  \left(  \alpha+2n-8-\Delta
_{E}\right)  }{8\alpha^{2}}\cos(\alpha t) \nonumber\\
&  -\frac{\Delta_{E}-2(n-4)+\alpha}{4\alpha}\cos\left[  t\frac{2n-\Delta
_{E}+\alpha}{2}\right] \nonumber\\
&  +\frac{\Delta_{E}-2(n-4)-\alpha}{4\alpha}\cos\left[  t\frac{2n-\Delta
_{E}-\alpha}{2}\right]  .\nonumber
\end{align}
Imposing $\Delta_{E}=2n$, PST is achieved for $t=\frac{1}{\alpha}\left(
2\pi+4\pi k\right)  $ with $k\in\mathbb{N}$.
\end{proof}

The main results are visualized in Fig. (\ref{fig1}) where the fidelity $f(i,j;t)$
between any two vertices $i$ and $j$ of a complete graph is plotted as a function
of time $t$. It is well known that, in this case, there is no PST without energy
shift as it is shown by the dashed line.  On the contrary, optimal energy shift
($\Delta_{E}=2n$) allows PST (fidelity equal to one) at times $t=\left( 2\pi+4\pi
k\right)/\alpha $ with $k\in\mathbb{N}$ as shown by the solid line.

\begin{figure}
[h]
\begin{center}
\includegraphics[
height=1.9207in,
width=2.8816in
]{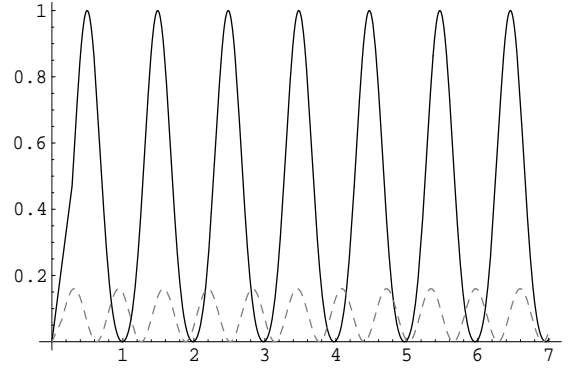}%
\caption{Fidelity $f(i,j;t)$ between any two vertices $i$ and $j$ of a complete graph with $n=5$ as a function of time $t$  in two settings: in the absence of energy shift (dashed line) and in the
presence of optimal energy shift (solid line).
\label{fig1}}
\end{center}
\end{figure}

\subsection{Complete graph with a missing link}

The graph $K_{n}^{-}$ is obtained from $K_{n}$ by deleting an edge, specifically
the one between the input and the output vertex. The next result describes the
behavior of the system in this case:

\begin{theorem}
\label{th2}Let $\beta=\sqrt{4(n^{2}+2n-7)-4(n-3)\Delta_{E}+\Delta_{E}^{2}}$ with
$n\geq4$ and $k\in\mathbb{N}$. For an energy shift  $\Delta_{E}(i,j)$ on the
vertices $i,j\in I/O$, we have the following observations:

\begin{itemize}
\item $\max_{t}f_{K_{n}}(i,i;t)=\max_{t}f_{K_{n}}(j,j;t)=1$, for $\Delta
_{E}(i,j)=2n-6$ and $t=2k\pi/\beta$;

\item $\max_{t}f_{K_{n}^{-}}(k,k;t)=1$, for every $k\notin I/O$ and
$t=4k\pi/\beta$.

When $i\neq j$,

\item $\max_{t}f_{K_{n}}(i,j;t)=1$, for $\Delta_{E}(i,j)=2n-6$ and $t=\left(
2\pi+4\pi k\right)/\alpha$;

\item $\max_{t}f_{K_{n}}(i,k;t)=16/\beta^{2}$, for $\Delta_{E}(i)=2n-6$,
$k\notin I/O$ and $t=\left(  2\pi+4\pi k\right)/\beta$;

\item $\max_{t}f_{K_{n}}(k,l;t)=[(\beta(n-2)-2)]^{2}/4\beta^{2}(n-2)^{2}$, for
$k,l\notin I/O$ and $t=2k\pi/\beta $.
\end{itemize}
\end{theorem}

\begin{proof}
The proof is very similar to the one of Theorem \ref{th1}. {The characteristic
polynomial $P(\lambda)$ of the }$XY$ adjacency matrix of $K_{n}^{-}$ {can be
obtained as function of $n$ and $\Delta_{E}$:
\begin{align*}
P(\lambda)  &  =(\lambda+2)^{n-3}(\Delta_{E}-2-\lambda)\times\left(
4(n-1)\right. \\
&  \left.  -2(n-3)\Delta_{E}+2(n-2)\lambda+\Delta_{E}\lambda-\lambda
^{2}\right)  .
\end{align*}
The roots of $P(\lambda)$ are as follows: }$\lambda_{1}=\Delta_{E}$,
$\lambda_{2}^{n-3}=-2$, $\lambda_{3,4}^{\pm}=(2(n-3)+\Delta_{E}\pm\beta)/2$. A
corresponding (unnormalized) orthogonal basis of eigenvectors can be written
as%
\begin{align*}
|\lambda_{1}\rangle &  =(-1,0,\ldots,0,1)\\
\lbrack|\lambda_{2}^{1\leq l\leq n-3}\rangle]_{u}  &  =\left\{
\begin{tabular}
[c]{rl}%
$-\frac{1}{l},$ & if $u\in\{2,n-r:1\leqslant r\leqslant l-1\};$\\
$1,$ & if $u=n-l;$\\
$0,$ & otherwise,
\end{tabular}
\ \ \ \right. \\
|\lambda_{3,4}^{\pm}\rangle &  =(1,\omega^{\pm},\ldots,\omega^{\pm},1),
\end{align*}
where $\omega^{\pm}=\frac{1}{4(n-2)}(2(n-3)-\Delta_{E}\pm\beta)$. The diagonal
entries of $U_{t}(K_{n})\equiv e^{-\iota H(K_{n})t}$ are given in terms of its
spectral decomposition:

\begin{itemize}
\item if $i\in I/O$ then%
\begin{align*}
\lbrack U_{t}(K_{n})]_{ii}  &  =\frac{1}{4\beta}\left(  \beta-2n+\Delta
_{E}+6\right)  e^{-\iota\left[  \lambda_{3}\right]  t}\\
&  +\frac{1}{4\beta}\left(  \beta-2n-\Delta_{E}+6\right)  e^{-\iota\left[
\lambda_{4}\right]  t}\\
&  +\frac{1}{2}e^{-\iota\left[  \lambda_{1}\right]  t};
\end{align*}

\item if $i\notin I/O$ then
\begin{align*}
\lbrack U_{t}(K_{n})]_{ii}  &  =\frac{1}{n-2}\left(  n-3\right)
e^{-\iota\left[  \lambda_{1}\right]  t}\\
&  +\frac{1}{2n\beta-4\beta}\left(  \beta-2n+\Delta_{E}+6\right)
e^{-\iota\left[  \lambda_{3}\right]  t}\\
&  +\frac{1}{2n\beta-4\beta}\left(  \beta+2n-\Delta_{E}-6\right)
e^{-\iota\left[  \lambda_{4}\right]  t}.
\end{align*}

\end{itemize}

The off-diagonal entries of $U_{t}(K_{n})$ are as follows:

\begin{itemize}
\item if $i\neq j$ and $i,j\in I/O$ then%
\begin{align*}
\lbrack U_{t}(K_{n})]_{ij}  &  =\frac{\Delta_{E}-2(n-3)+\beta}{4\beta
}e^{-\iota\left[  \lambda_{3}\right]  t}\\
&  +\frac{2(n-3)-\Delta_{E}+\alpha}{4\beta}e^{-\iota\left[  \lambda
_{4}\right]  t}\\
&  -\frac{1}{2}e^{-\iota\left[  \lambda_{1}\right]  t};
\end{align*}

\item if $i\neq j$, $i\in I/O$ and $j\notin I/O$ or \emph{viz}., then
$[U_{t}(K_{n})]_{ij}=2(e^{-\iota\left[  \lambda_{3}\right]  t}-e^{-\iota
\left[  \lambda_{4}\right]  t})/\beta$;

\item if $i\neq j$ and $i,j\notin I/O$ then%
\begin{align*}
\lbrack U_{t}(K_{n})]_{ij}  &  =\frac{\Delta_{E}-2(n-3)+\beta}{2(n-2)\beta
}e^{-\iota\left[  \lambda_{4}\right]  t}\\
&  +\frac{2(n-3)-\Delta_{E}+\alpha}{2(n-2)\beta}e^{-\iota\left[  \lambda
_{3}\right]  t}\\
&  -\frac{1}{(n-2)\beta}e^{-\iota\left[  \lambda_{2}\right]  t}.
\end{align*}

\end{itemize}

If we take $i,j\in I/O$, the fidelity $f(i,j;t)=|\langle j|U_{t}(K_{n}%
)|i\rangle|^{2}$ reads%
\begin{align}\label{fidml}
f(i,j;t)  &  =\frac{\Delta_{E}^{2}+3\beta^{2}-4\Delta_{E}(n-3)+4(n-3)^{2}%
}{8\beta^{2}}\\
&  +\frac{(6+\Delta_{E}+\beta-2n)(\beta+2n-6-\Delta_{E})}{8\beta^{2}}%
\cos(\beta t) \nonumber\\
&  -\frac{6-2n+\Delta_{E}+\beta}{4\beta}\cos\left[  t\frac{2n-6-\Delta
_{E}+\beta}{2}\right] \nonumber\\
&  +\frac{6-2n+\Delta_{E}-\beta}{4\beta}\cos\left[  t\frac{2n-6-\Delta
_{E}-\beta}{2}\right] \nonumber
\end{align}
If we assume that $\Delta_{E}=2n-6$ then PST is achieved for $t=\frac{1}%
{\beta}\left(  2\pi+4\pi k\right)  $ with $k\in\mathbb{N}$.
\end{proof}

\bigskip

Notice that for $K_{n}^{-}$ we have $\Delta_{E}=2n-6$, while for $K_{n}$ we have
$\Delta_{E}=2n$. This fact alone does not provide enough information to
conjecture that the energy shift required for PST in a graph with $m$ edges is
proportional to $m$. Indeed, the energy shift appears to be a nonlinear function of
the eigensystem of the matrix $H_{XY}(G)$. Also, notice that the matrices
$H_{XY}(K_{n})$ and $H_{XY}(K_{n}^{-})$ satisfy the relation $(H_{XY}%
(K_{n})\cdot H_{XY}(K_{n}^{-}))^{T}=H_{XY}(K_{n}^{-})\cdot H_{XY}(K_{n})$.

\begin{figure}
[h]
\begin{center}
\includegraphics[
height=1.9207in,
width=2.8816in
]{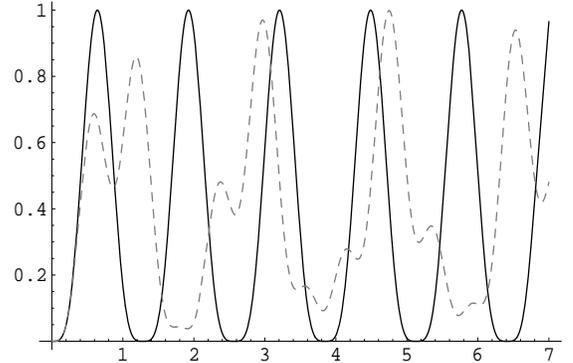}
\caption{Fidelity $f(i,j;t)$ between the two nonadjacent vertices $i$ and $j$ of  a complete graph with a missing link with $n=5$ as a function of time $t$ in two settings: in the absence of energy shift (dashed line) and in the
presence of optimal energy shift (solid line).
\label{fig2}}
\end{center}
\end{figure}

The main results are visualized in Fig. (\ref{fig2}). Here, the fidelity $f(i,j;t)$
between the two nonadjacent vertices $i$ and $j$ of $K_{5}^{-}$ is plotted as a
function of time $t$. In this case, it is known \cite{ca} that for $n$ multiple of four
there is PST in the isotropic Heisenberg model without energy shift. In the $XY$
model considered here, numerical solutions of the suitable \emph{t} in Eq.
(\ref{fidml}) suggest that \emph{PST} can be reached also in other cases as
shown by the dashed line in Fig. (\ref{fig2}).  In this case the fidelity is one for
$t\approx 5$ and $n=5$. This is a remarkable generalization that we conjecture
true for every $n$. An intuition of this fact can be derived substituting the value of
$\Delta_{E}=0$ in the Eq. (\ref{fidml}) that turns out to be independent of the
number of nodes $n$ when the fidelity reaches its maximum. However, the use
of optimal energy shift ($\Delta_{E}=2n-6$) allows PST at shorter times
$t=\frac{1}{\beta}\left( 2\pi+4\pi k\right) $ with $k\in\mathbb{N}$ as shown by the
solid line, thus facilitating the information transfer.

\section{Spin chains}

A special case of a network is represented by a linear spin chain, where the
vertices at the extremities are considered as input and output. By adding an
appropriate free energy shift $\Delta_{E}$ (independent of the number of nodes
$n$) to these vertices, numerical results involving relatively large chains point out
that we can always achieve PST. This is remarkable because usual spin chains
with more than three vertices do not allow PST. As a counterpart of this fact, the
transfer time generally grows rapidly with the number of vertices and with the
amount of energy $\Delta_{E}$. However, in some special cases like the three
vertices chain, where PST is achievable without adding energy, the energy shift
only causes a larger transfer time. Table \ref{tbl} shows accordingly some
numerical results for chains of small length and energy shift on the end-vertices.
Apart from $n=2$, for the sake of clarity, we take the closest integer to the real
values obtained.

\begin{table}
  \centering
\begin{tabular}
[c]{c|cccc}%
$\Delta_{E}\backslash n$ & $2$ & $3$ & $4$ & $5$\\\hline
$10$ & $0.7$ & $5$ & $19$ & $99$\\
$20$ & $0.7$ & $8$ & $81$ & $8010$\\
$30$ & $0.7$ & $12$ & $178$ & $2665$\\
$40$ & $0.7$ & $16$ & $313$ & $6260$\\
$50$ & $0.7$ & $20$ & $494$ & $12294$%
\end{tabular}
\caption{Numerical results of the transfer time for chains of small length and
varying energy shift on the end-vertices.}\label{tbl}
\end{table}

It is plausible that the transfer time in a generic network decreases as the
number of paths between the input and the output vertex increases. The
minimum transfer time is clearly achieved when the two vertices are adjacent.
Thus, for a fixed amount of energy $\Delta_{E}$, numerical results show that the
transfer time for the maximum fidelity, $t_{ij}(G)$, between vertices $i$ and $j$
of a network $G$ on $n$ vertices, is%
\begin{equation}
t_{ij}(G)\approx\mathcal{O}\left(  \frac{t_{\Delta_{E},k}}{p_{min}%
(i,j)}\right)  . \label{vpk}%
\end{equation}
Here $t_{\Delta_{E},k}$ is the transfer time of the spin chain with $n$ vertices
and $p_{min}(i,j)$ is the number of different geodesics between $i$ and $j$. Table
\ref{Table1a} shows the transfer time required to obtain a fidelity close to one,
when we consider antipodal vertices in graphs of a family constructed as follows:
only two vertices, which are then said to be \emph{antipodal}, have degree $l$;
all other vertices have degree $2$ and belong to paths connecting the antipodal
vertices. Such paths are disjoint and have only the antipodal vertices in common.
The number of vertices in a graph with $l$ paths of length $n$ is $n+\left(
n-2\right)  l$, for $n\geq3$. Table \ref{Table1a} gives evidence that we can
gradually cut the transfer time by increasing the number of paths. Intuitively, an
equivalent result should be also obtained by modifying the couplings in the
original chain.

\begin{table}
  \centering
\begin{tabular}
[c]{c|ccc}%
$l\backslash n$ & $3$ & $4$ & $5$\\\hline
$1$ & $5$ & $19$ & $99$\\
$2$ & $3$ & $11$ & $62$\\
$3$ & $2$ & $9$ & $42$\\
$4$ & $1$ & $6$ & $36$%
\end{tabular}
\caption{Numerical results for the decrease of the transfer time with a fixed
energy shift on the end-vertices and an increasing number of
paths}\label{Table1a}
\end{table}

\section{Fluctuations}

In this section, we analyze the problem of transferring an energy excitation in the
presence of noise. We keep working with $K_{n}$ and $K_{n}^{-}$. In practice,
we consider a gaussian stochastic process $\xi_{ij}$ of zero mean and
$\sigma^{2}$ variance, affecting the energy of the particles (qubits' frequencies)
or the interaction energies (qubits' couplings). Under this assumption, the
Hamiltonian entries become
\[
\lbrack H_{XY}(G,\xi)]_{ij}=\left\{
\begin{tabular}
[c]{ll}%
$\Delta_{E}+\xi_{ii},$ & if $i=j\in I/O;$\\
$2+\xi_{ij},$ & if $ij\in E(G);$\\
$0+\xi_{ij},$ & otherwise.
\end{tabular}
\right.
\]
We then distinguish two cases: noise affecting the vertices and noise affecting
the edges. Formally,

\begin{itemize}
\item[1.] $\xi_{ii}\neq 0$, $for every i \in V(G)$ and $\xi_{ij}=0$, when $i\neq j$;

\item[2.] $\xi_{ij}\neq 0$, $ for every i and j\in E(G)$ and $\xi_{ii}= 0$.
\end{itemize}

We are interested in evaluating the average fidelity as a function of the variance
of the independent gaussian random variables. The chosen energy shift
$\Delta_{E}$  is the optimal one, according to the results of Section III.

\begin{figure}
[h]
\begin{center}
\includegraphics[
height=3.003in,
width=3.5in
]{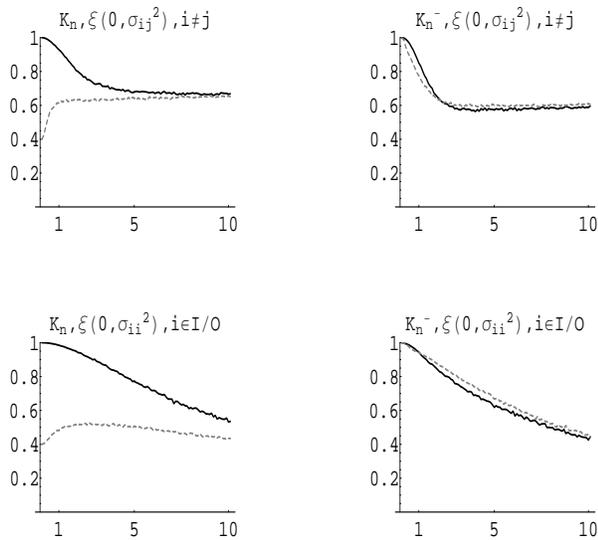}
\caption{On the left, average fidelity between any two vertices of a complete graph as a function of the variance $\sigma^{2}$ (at optimal time).
At the bottom (resp. top) is represented the case 1. (resp. 2.). The solid line refers to the
presence of optimal energy shift at input/output vertices while the dashed line refers to the
absence of such shift.
On the right, average fidelity between the two nonadjacent vertices of a complete graph with a missing link as a function of the variance $\sigma^{2}$ (at optimal time).
At the bottom (resp. top) is represented the case 1. (resp. 2.). The solid line refers to the
presence of optimal energy shift at input/output vertices while the dashed line refers to the
absence of such a shift.
\label{fig3}}
\end{center}
\end{figure}

The results are reported in Fig.\ref{fig3}. By comparing top and bottom graphics
we can see that disordered couplings are more deleterious than disordered
frequencies with an optimal energy shift. This fact  has been already pointed out
in a different context by Gammaitoni \emph{et al.} in \cite{Fazio}. The decay of
fidelity over $\sigma^{2}$ comes from the fact that the noise causes localization
phenomena for the excitation transfer \cite{And}. This is more evident for
$K_{n}^{-}$ where the ``degree of disorder" is higher (compare top-left and
top-right plots). Furthermore, in the absence of energy shift the noise may
enhances the transmission fidelity (top-left and bottom-right plots). This is
reminiscent of stochastic resonance effects \cite{Fabio}.

\section{Conclusions}

We have shown how to enhance the fidelity of excitation transfer in a quantum
spin network with a fixed interaction ($XY$ model) and network. It turns out that
it is possible to achieve perfect transfer with the use of suitable energy shifts in
all-to-all networks and in all-to-all networks with a missing link. We conjecture
that this is possible in \emph{any} network. This technique is promising for future
applications as recent works with super-conducting qubits suggest (\cite{st} and
references therein). Finally, we have shown how different kinds of noise affect the
transfer fidelity. We believe that our results could open up new perspectives for
communication or information processing in quantum networks.

\bigskip

\noindent\emph{Acknowledgments.} The authors would like to thank Stefano
Pirandola,  Masoud Mohseni and Yasser Omar  for useful discussions. The work
of S. M. is supported by the European Commission, under the FET-Open grant
agreement HIP, number FP7-ICT-221889. Research at IQC is supported in part
by DTOARO, ORDCF, CFI, CIFAR, and MITACS.


\begin{thebibliography}{99}                                                                                               %


\bibitem {And}Anderson P. W., \emph{Phys. Rev.} \textbf{109}, 1492 (1958).

\bibitem {be}Bernasconi A., Godsil C., Severini S., \emph{Phys. Rev. A}
\textbf{78},052320 (2008). arXiv:0808.0510v1 [quant-ph].

\bibitem {sou}Bose S., \emph{Phys. Rev. Lett.} \textbf{91}, 207901 (2003). arXiv:quant-ph/0212041v2.

\bibitem {sou1}Bose S., \emph{Contemporary Physics,} \textbf{Vol. 48} (1), pp.
13-30, 2007. arXiv:0802.1224v1 [cond-mat.other].

\bibitem {ca}Bose S., Casaccino A., Mancini S., Severini S., Communication in
XYZ All-to-All Quantum Networks with a Missing Link, \emph{Int. J. Quantum
Inf.}, \textbf{7}, no 4, 2009. arXiv:0808.0748v2 [quant-ph].

\bibitem {ch}Christandl M., Datta N., Ekert A. and Landahl A.J., \emph{Phys.
Rev. Lett.} \textbf{92}, 187902 (2004). arXiv:quant-ph/0309131v2.

\bibitem {Fazio} De Chiara G., Rossini D., Montangero S. and Fazio S.,
\emph{Phys. Rev. A} \textbf{72}, 012323 (2005). arXiv:quant-ph/0502148v2.

\bibitem {Fabio}Gammaitoni L., Hanggi P., Jung P. and Marchesoni F.,
\emph{Rev. Mod. Phys.} \textbf{70}, 223 (1998).

\bibitem {Kimble}Kimble H.J., \emph{Nature} \textbf{453}, 1023 (2008).
arXiv:0806.4195v1 [quant-ph].

\bibitem {st}Strauch F. W. and Williams C. J., \emph{Phys. Rev. B} \textbf{78},
094516 (2008). arXiv:0708.0577v3 [quant-ph].

\bibitem {sub}Subrahmanyam V., \emph{Phys. Rev. A} \textbf{69}, 034304 (2004). arXiv:quant-ph/0307135v2.
\end{thebibliography}
\end{document}